\documentclass[single,letter]{sig-alternate-per}

\usepackage{graphicx}

\usepackage{graphicx}

\makeatletter
\numberwithin{equation}{section}
\numberwithin{figure}{section}
\newtheorem{thm}{Theorem}
  \newtheorem{lem}[thm]{Lemma}
\newtheorem{proposition}[thm]{Proposition}
 
\makeatother

\begin{document}

\title{Pricing Agreement between Service and Content Providers: A Net Neutrality
Issue}

\numberofauthors{3}

\author{
\alignauthor
Alexandre  REIFFERS-MASSON\\
     \affaddr{INRIA / LIA,}\\
      \affaddr{2004 Route des Lucioles, 06902 Sophia-Antipolis, France}\\
     \email{alexandre.reiffers@inria.fr}
\alignauthor Yezekael HAYEL\\
 	\affaddr{LIA, University of Avignon, }\\
 	\affaddr{339, chemin des Meinajaries, 84911 Avignon, France}\\
 	\email{yezekael.hayel@univ-avignon.fr}\\
\alignauthor Eitan ALTMAN\\
\affaddr{INRIA}\\
      \affaddr{2004 Route des Lucioles, 06902 Sophia-Antipolis, France}\\
 	\email{eitan.altman@inria.fr}
}

\maketitle

\begin{abstract}
The Net Neutrality issue has been at the center of debate worldwide lately. Some countries have established laws so that principles of Net Neutrality are respected, the Netherlands being the latest country in Europe. Among the questions that have been discussed in these debates are whether to allow agreements between service and content providers, i.e. to allow some preferential treatment by an operator to traffic from some subset of providers. Our goal in this paper is to analyze the impact of non neutral pricing and agreements on the Internet users and on the content providers. Each one of several Internet users have to decide in which way to split their demand among several content providers. The cost for an Internet user depends on whether the content providers have an agreement with the Internet Service Provider in which the Internet user is connected to. In addition, the requests coming from users depend on the preference of the consumer in the different CP. We model the choice of how to split the demands and the pricing aspects faced by the content providers as a hierarchical game model composed of a congestion game at the lower level and a noncooperative pricing game at the upper level. We show that agreement between providers have a positive impact on the equilibrium performance of the Internet users. We further show that at equilibrium, the content provider on the contrary, does not benefit from the agreement.
\end{abstract}

\section{Introduction}

Network Neutrality is an approach to provide network access without unfair discrimination between applications, content or specific source of the traffic.
If there are two applications or services or providers that require the same network resources and one is offered better quality of service (delays, speed, etc.) or is cheaper to access, then there is a discrimination. Although Net neutrality concerns many aspects related to discrimination, there is one particular economic issue that is at the heart of the conflict over network neutrality: the relation between service and content providers. Hahn and Wallsten \cite{key-5} write that net neutrality "usually means that broadband service providers charge consumers only once for Internet access, do not favor one content provider over another, and do not charge content providers for sending information over broadband lines to end users.'' This motivated the study in \cite{key-6}. In this paper we study the impact of other aspects of non-neutrality that arises in the relations between service and content providers, i.e. the possibility that an ISP gives preference or pricing agreement with a content provider. In some industries, laws against vertical monopolies are enforced which in some cases obliged companies to split their activity into separate specialized companies; this was the case of railways companies in Europe which were obliged to separate their rail infrastructure and the service part of the activity which concerns public transportation by trains. In contrast, in the telecom market, the same company may propose both the networking services and content or an ISP and a CP can have a "pricing agreement". The aim of this paper is to study the implication of such economic relationships between providers on the Internet users. Specifically, we try to find a good answer to the next question: Is a "Pricing agreement" between an ISP and a CP good for subscribers? We are suggesting here a new point of view of a "pricing agreement". Usually in a Net Neutrality issue , the problem of agreement or disagreement between ISP and CP is a vertical foreclosure (Degradation of traffic) \cite{key-11}. This type of problem has been observed in France between Free (a French ISP) and Google \cite{key-12}. In our paper, if a CP and an ISP have a "pricing agreement" then a subsciber of the ISP mentionned above doesn't have to pay for the access to this CP's content as illustrated in \cite{key-10}.

We introduce a two level hierarchical game where we consider two types of competition. The first one comes from the subscribers through routing decisions by determining the sources of the content they download. The interaction occurs through the preference induced by these decisions. The second type of competition is between the content providers (CP) through the price they ask for downloading content from their servers. We first study the general noncooperative game where each decision maker is modeled as a selfish player. We assume a two level game based on the prices proposed by the content providers, the subscribers determine selfishly their demand. In a second approach, we consider that each service provider can make an agreement with a content provider, and then the subscribers of this service provider can have access to contents of this particular content provider free of charge. We study the impact of this agreement on the equilibrium of the market.

We show that agreements between providers have a positive impact on the equilibrium performance (cost perceived by users) of the Internet users. We further show that at equilibrium, on the contrary, the content provider does not benefit from the agreement with the service provider, in terms of their revenue.

Related works: Ozdaglar and Acemoglu \cite{key-6}, study the game between subscribers, but in their case subscribers play a non atomic selfish routing game \cite{key-7}. In our case we consider an atomic selfish routing game where a player is modeled as a source of splittable traffic. After all, in \cite{key-6} Ozdaglar and Acemoglu, CPs control both flows and prices, while in our work CPs control only prices; and each Internet user can determine the source of the traffic that he wishes to download. 

The paper is organized as follows. In section \ref{model} we present the general economic model and we describe the hierarchical game proposed. We study it in section \ref{equi1} for the general case without agreements between the providers. In section \ref{equi2} we investigate the agreement between each service provider and one content provider. We compare the cost for an Internet user and the revenue for a content provider to the general hierarchal game context in section \ref{comp}. Finally, in section \ref{conc}, we summarize the results that we have obtained in this paper and we give some perspectives.

\section{General Economic Model}
\label{model}

In the actual Internet market, we can find examples of pricing agreements between CPs and IPSs. For example, Orange is a French ISP and Deezer is a French music streaming service (a content provider). According to the Financial Times \cite{key-10}, "As part of the deal with Deezer, Orange will make available a special mobile-only tariff for pay-monthly customers, to avoid the £9.99 standalone cost of Deezer’s top package." Therefore, customers with an Internet subscription with Orange will have preferential offer for listening music in the website deezer (see description in figure \ref{fig1}.1). 

\begin{figure}[!ht]
\center
\includegraphics[width=8.5cm]{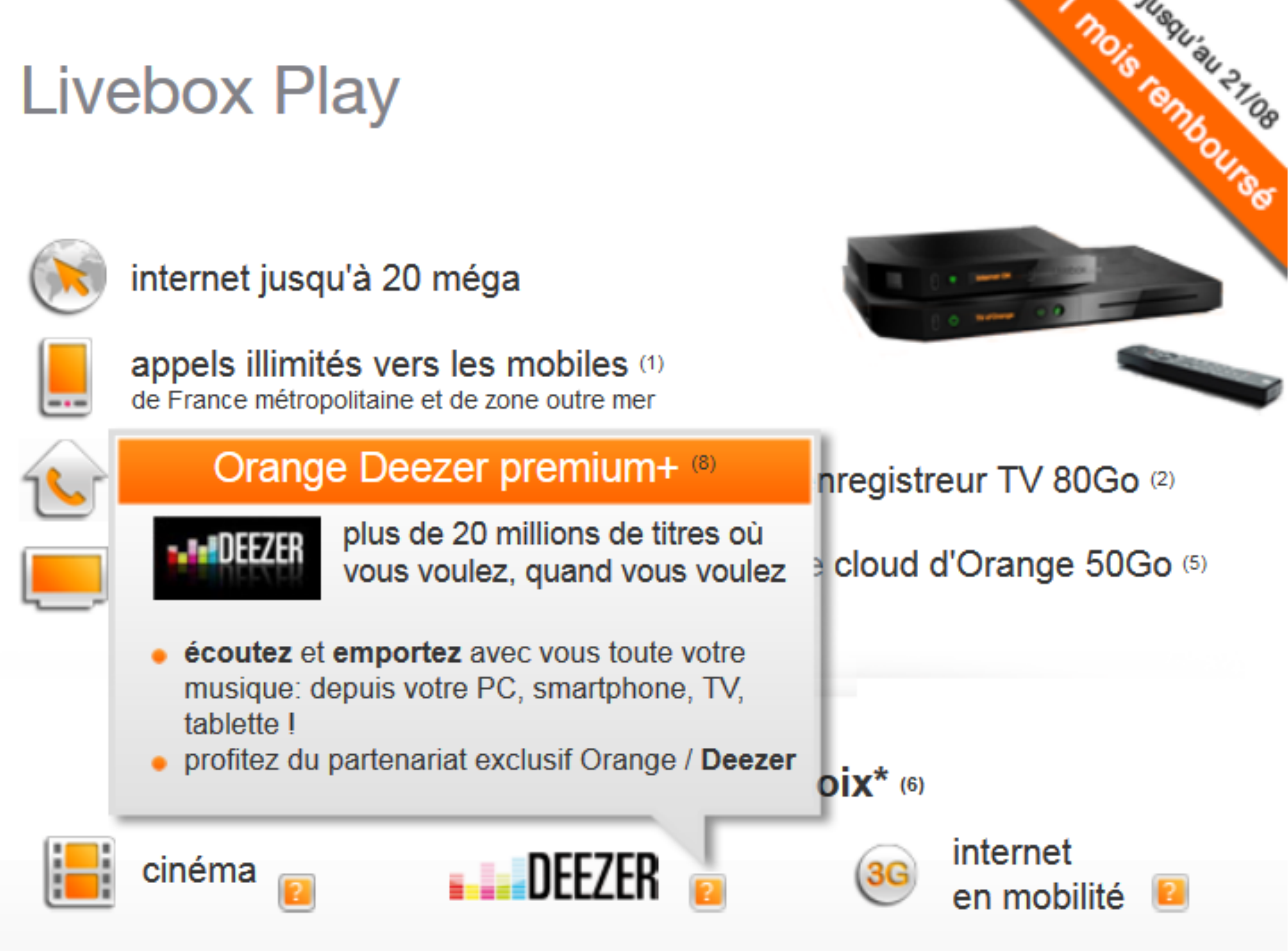}
\label{fig1}
\caption{Example of agreement between an ISP and a CP.}
\end{figure}

We consider a general economic model of content distribution over the Internet. We assume that several content providers, say $M$, are able to broadcast some contents over the Internet. The traffic is carried through the network by high level Internet Service Providers (called ISP) which have direct links to all CP. Finally, those ISP are connected to local ISP, denoted subscribers, which distribute the content to a mass of end users. We consider in our system $M$ CP and $N$ ISP. Each subscriber is connected to an ISP and cannot access directly any CP. We consider in our study the download traffic; contents on the Internet are generated from servers to clients (video-on-demand, movie broadcast, etc). The source of each content flow is a CP and the corresponding destination is a set of subscribers.

\begin{figure}[!ht]
\center
\includegraphics[width=8.5cm]{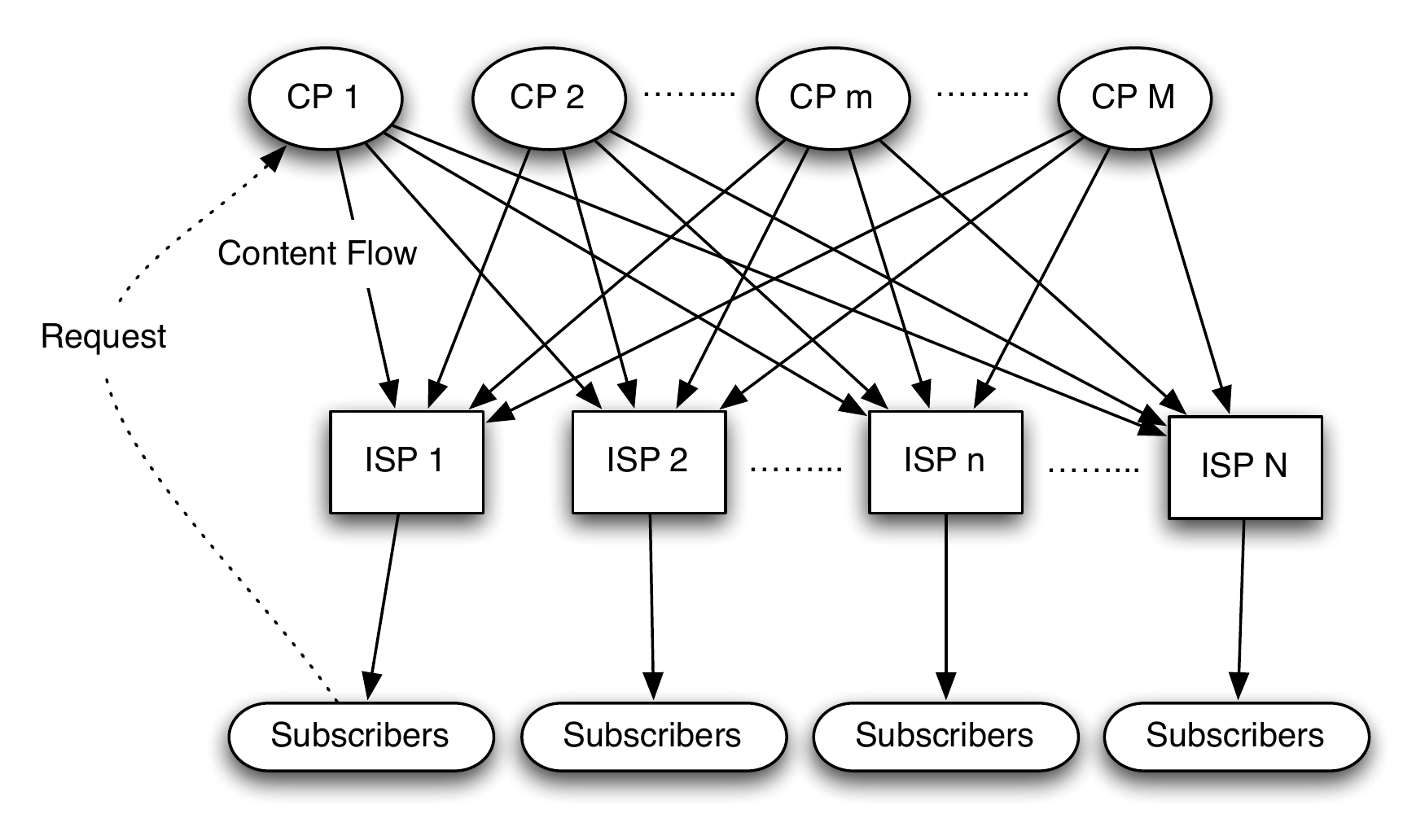}
\label{fig2}
\caption{Economic model with service and content providers.}
\end{figure}

We consider that each ISP $n$ for $n \in \{1,\ldots,N\}$ receives some demands from its subscriber about an access to some contents. Requests are sent from each subscriber $n$ to all CP $m$.  At each CP, any request from a subscriber $n$ (i.e. set of subscribers connected to ISP $n$) induces a traffic rate $x_n^m>0$ from a CP $m$, to the ISP $n$, which is therefore aggregated and sent to the subscriber connected to ISP $n$.  We assume that the total traffic flow from ISP $n$ to its subscriber is $\phi_n$. The network is symmetric and the traffic $\phi_n$ is the same for all subscribers and CP $n$, i.e. $\forall n, \quad \phi_n=\phi$. This demand $\phi$ is an average value of the total amount of requests that all ISPs receive. The economic relationships of our model are described in figure \ref{fig2}.2. \\

\subsection{Congestion game between subscribers}

We consider a noncooperative routing game where the decision of subscriber connected to ISP $n$, is how to split his download traffic $\phi$ from all the CP, i.e. the decision variables for subscriber connected to ISP $n$ is the vector $x_{n}=(x^1_{n},x^2_{n},\ldots,x^M_{n})$. We denote by $p^m$ the charge, per unit of traffic, that  a subscriber has to pay in order to download traffic from CP $m$. Then, for a traffic quantity $x^m_{n}$, the subscriber $n$ has to pay $x^m_{n}p^m$ to the CP $m$. We consider also a preference cost at each CP. Indeed, if we assume that each CP has the same quality of service, a subscriber prefers to download his content from the less crowded CP. This cost depends on the total download traffic generated at each CP $m$, that is $\displaystyle{\underset{n}{\sum}}x^m_{n}$. Let $D^m:\mathbb{R^+}\mapsto\mathbb{R^+}$ be the preference cost function at CP $m$ which we assume to be convex and increasing. The preference cost perceived by a subscriber connected to ISP $n$ and downloads a traffic $x_n^m$ from CP $m$, is equal to $x_{n}^m\cdot D^m(\displaystyle{\underset{n}{\sum}}x^m_{n})$. Then the total cost (content price + congestion cost) for a subscriber connected to ISP $n$ is given by:

\label{cost}
$$
C_{n}(\mathbf{x_{n}},\mathbf{x_{-n}},\mathbf{p})=\displaystyle{\underset{m}{ \sum}}x_{n}^{m}\left[D^m(\displaystyle{\underset{n}{\sum}}x_{n}^m)+p^{m}\right]
$$ 
where $\mathbf{x_{n}}= \{x^1_{n},\ldots,x^M_{n}\}$ is the decision vector for the subscriber connected to ISP $n$, $\mathbf{x_{-n}}$ is the decision vectors of all the other subscribers connected to ISP $n$. $\mathbf{p}=\{ p^1,\ldots,p^M\}$ is the price vector of all the CP. Given this price vector, each subscriber will minimize his cost function under his demand constraint:

$$\forall n, \quad \min\;C_n(\mathbf{x_{n}},\mathbf{x_{-n}},\mathbf{p})$$

$$
\text{s.t.}\quad \underset{m=1}{\overset{N}{\sum}}x_{n}^{m}=\phi. 
$$
\label{optcost}

We then have a first noncooperative congestion game between the subscribers, as they interact through the congestion at each CPs, given the prices determined by the content providers. Therefore we analyse our economic model considering the routing game depicted in figure 2.3 which is equivalent to \ref{fig2}.2 We consider therefore, in a second step, that the CPs choose optimally their tariff $p^m$ in order to optimize their own revenue. This decision will impact the equilibrium between the subscribers and then we are faced with a hierarchical structure between the decision makers, namely the subscribers and the CPs. In order to solve such hierarchical game, we first study the equilibrium of the subscribers (the followers) for a given decision vector of the CPs (the leaders). Secondly, we consider this underlying equilibrium in order to determine the equilibrium between the CPs. 

\begin{figure}[!ht]
\center
\includegraphics[width=8.5cm]{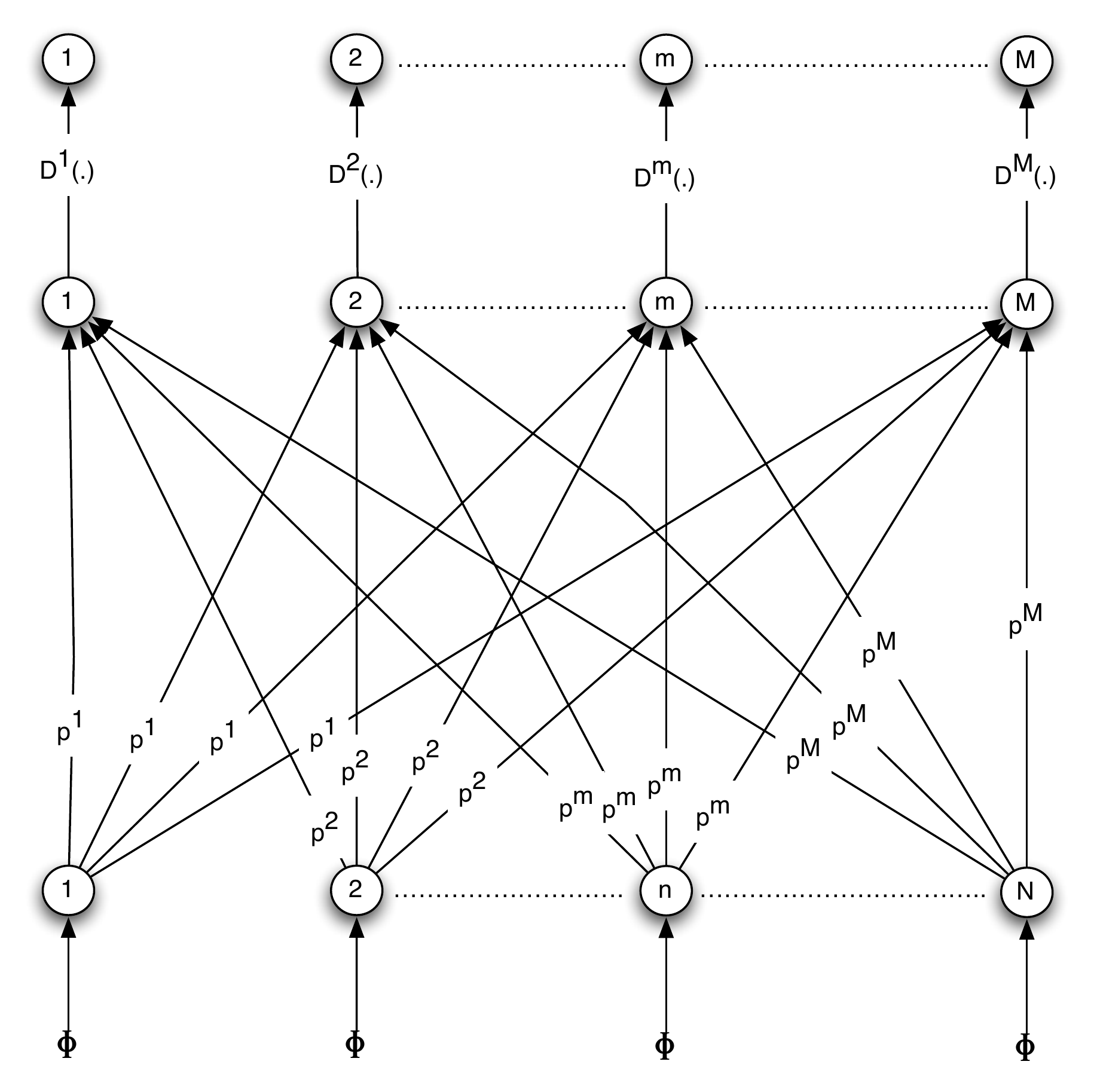}
\label{fig3}
\caption{Equivalent routing game in the case of non agreement}
\end{figure}

\subsection{Content provider game}
The revenue of CP $m$ is defined by:
$$
\Pi(\mathbf{\underline{x}},p^m,p^{-m})=p^{m}\underset{n}{\sum}\underline{x}_{n}^{m}(\textbf{p}),
$$
where $\underline{x}_{n}^{m}(\textbf{p})$ is the equilibrium traffic flow from CP $m$ to ISP $n$, $p^m$ is the price charged by CP $m$ and $p^{-m}$ is the price vector of all the other CP. We assume that the decision variable of each CP $m$ is in an interval, i.e.  $p^m \in [0,p_{\text{max}}]$. The system is totally symmetric, in the sense that the quantity of traffic $\phi$ is the same for all subscribers $n$ and preference cost functions do not depend on $m$, i.e. $\forall m, \quad D^m(.)=D()$. Based on this symmetry property of the game, we can use results of \cite{key-3} and assume the existence of a symmetric equilibrium for our hierarchical game:
$$
\forall m, \quad \max_{p^m}\;p^{m}\underset{n}{\sum}\underline{x}_{n}^{m}(\textbf{p})
$$
such that
$
\forall n$, $\underline{x}_{n}^{m}(\textbf{p})$  is an equilibrium for all the subscribers. First, this symmetric assumption can be justified by the fact that in a large network, we can approximate the behavior of many end users with only one user which has the average characteristics of all the end users. Secondly, this assumption allows us to obtain explicit formulations, described in the next section, of the equilibrium of our complex hierarchical game.

\section{Equilibrium analysis}
\label{equi1}
We are faced with a non-cooperative pricing game considering an underlying congestion game. In order to determine the equilibrium of this hierarchical game, we first solve the congestion game between the subscribers (for a fixed price vector of the CPs). Secondly, we solve the non-cooperative pricing game between the content providers taking into account the underlying equilibrium between the subscribers. We consider linear cost function $D(x)=ax$ as in \cite{key-6}.\\

\begin{proposition}
In this game with independent service providers and content providers, a unique symmetric equilibrium $(x_n^{m},p^m)= (\underline{x},\underline{p})$ exists for all $(n,m)\in\{1,\ldots,N \}\times \{1,\ldots,M\} $, given by:
$$
\underline{x}^m_n=\frac{\phi}{M}\quad \mbox{and}\quad \underline{p} =(N-1)\phi a.
$$\\
\end{proposition}

\begin{proof}

 First let $$L_{n}(\mathbf{x_{n}},\mathbf{x_{-n}},\mathbf{p},\lambda_{n})=$$
 $$\displaystyle{\underset{m}{ \sum}}x_{n}^{m}\left[D^m(\displaystyle{\underset{n}{\sum}}x_{n}^m)+p^{m}\right]-\lambda_{n}(\displaystyle{\underset{m}{\sum}}x_{n}^m-\phi) $$ the Lagrangian function associated to the cost function $C_{n}(\cdot)$. We look for a symmetric equilibrium between the CPs, i.e. for the noncooperative pricing game at the upper layer. Then we assume that CPs $m'\in \{1,\ldots,M \}-\{m\}$ play $q$ and one CP, say $m$, plays $p^m$. We want to find some $q$ where the best reply of CP $m$ against $q$ is $q$. First, we have to determine the equilibrium flows between the subscribers, depending on those prices,i.e. $\mathbf{\underline{x}}(p^m,q)$ for all $p^m$ and $q$. We look for $\mathbf{\underline{x}}(p^m,q)$, a solution of the following system: 

 $$ \begin{cases}
{\displaystyle \frac{\partial L_{n}}{\partial x_{n}^{m}}}(\mathbf{x_{n}},\mathbf{x_{,-n}},p^m,q,\lambda_{n})=0\\
\underset{m}{\sum}x_{n}^m=\phi,\;\forall(n,m)
\end{cases}.$$

This game has a strong symmetric property as all subscribers are interchangeable. Then, we can restrict ourselves to two strategies $x$ and $y$ where $x$ is a request for CP $m'$ and $y$ is for CP $m$. This induces a great simplification in the analysis of our complex hierarchical game. Figure \ref{fig4}.1 shows the different variables of our system. Thanks to \cite{key-3}, the previous system is equivalent to the following one:\\

$$\begin{cases}
x\displaystyle{\frac{\partial D}{\partial  x}}(Nx)+D(Nx)+q-\lambda=0\\
\\
y\displaystyle{\frac{\partial D}{\partial y}}(Ny)+D(Ny)+p^m-\lambda=0\\
\\
(M-1)x+y=\phi.
\end{cases}$$

We denote by $\underline{x}$ and $\underline{y}$ a solution of this system. We consider the linear cost function $D(x)=ax$ and thus we have to solve the following linear system:

$$\begin{cases}
ax+a(Nx)+q-\lambda=0\\
\\
ay+a(Ny)+p^m-\lambda=0\\
\\
(M-1)x+y=\phi.
\end{cases}$$  \\

 If $\underline{x} <0$ (or $\underline{y}<0$) then  $\underline{x} =0$ (or $\underline{y}=0$). And if $\underline{x} >\phi$ (or $\underline{y}>\phi$) then  $\underline{x}=\phi$ (or $\underline{y}=\phi$).\\

Let's now consider the CP $m$ and how it's going to optimize its revenue, which is the function $R^m(p^m,q)=p^m \underline{y}N$. Its best reply against all other CPs that play $q$ is given by $p^m$ solution of $\displaystyle{\frac{\partial R^m}{\partial p^m}}(p^m,q)=0$. We have to find a certain $p$ which is a solution of $\displaystyle{\frac{\partial R^m}{\partial p^m}}(p,p)=0$. We denote this equilibrium by $\underline{p}$. Considering the linear cost function, we obtain:

$$\underline{p} =(N-1)\phi a .$$

 If $\underline{p}>p_{max}$ then $\underline{p}=p_{max}$. We have a particular interest in the case where $p_{max}>(N-1)\phi a$. Then the equilibrium flow from CP $n$ to ISP $m$ is \small $\underline{x}^m_n=\displaystyle{\frac{\phi}{M}}$.\\

\begin{figure}[!ht]
\center
\includegraphics[width=5cm]{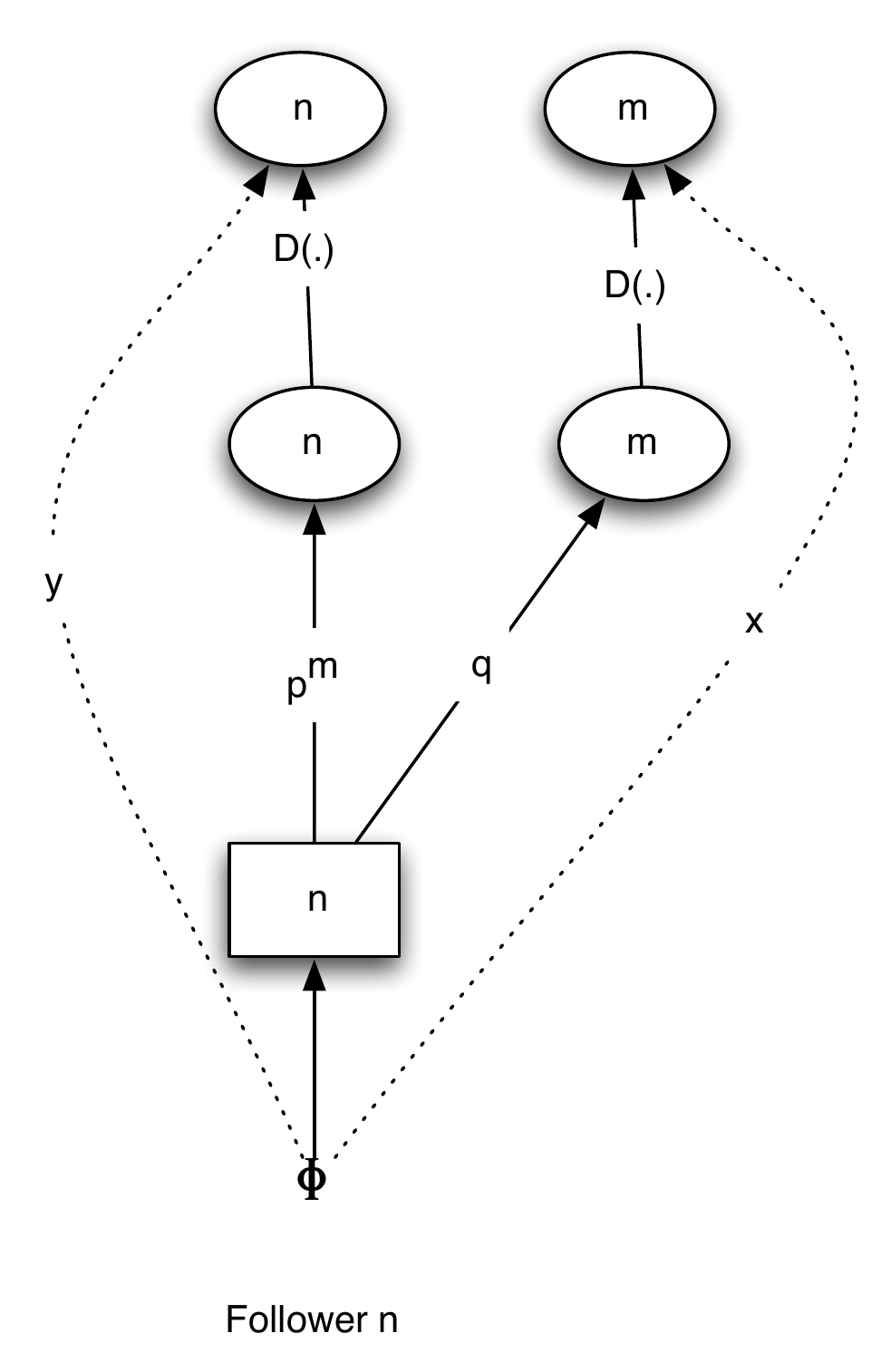}
\label{fig4}
\caption{Followers symmetric strategies (help for the proof)}
\end{figure}

%
%
%
%
%
%

\end{proof}

Considering this result, we are able to determine the cost for a subscriber and the revenue of a CP at equilibrium. In fact, this proposition gives the equilibrium prices and the value of the total traffic generated by each CP at equilibrium. Note that, as defined in \cite{key-2}, we also have uniqueness of this total traffic at equilibrium:
$$
\forall m, \quad \sum_{n}^{}\underline{x}^m_n=\frac{N\phi}{M}.
$$
The cost for a subscriber connected to CP $n$ at the equilibrium is given by:
$$
C_{n}(\underline{x},\underline{p})= \phi^2a(N-1+\frac{N}{M}).
$$ 
The revenue for any CP $m$ is:  
$$
\Pi_m(\underline{p})=\frac{(N-1)N}{M}\phi^2a.
$$

\section{Agreements between service and content providers}
\label{equi2}
We consider now that each ISP $n$ makes an agreement with a CP. Then, in order to have a symmetric configuration such as each ISP has an agreement with one CP and vice-versa, we assume that the number of ISP is equal to the number of CP, i.e. $M=N$. In order to simplify, $n$ is the index of the CP which has an agreement with ISP $n$.  These agreements imply that the charge $p^n$ is equal to $0$ for the  traffic generated from the CP $n$ to the ISP $n$ (see the routing game \ref{fig5}.1). Then, the total cost for the subscriber connected to ISP $n$ becomes:
$$
C^v_{n}(\mathbf{x_{n}},\mathbf{x_{-n}},\mathbf{p})=
\underset{m\neq n}{ \sum}x_{n}^{m}\left[D(\underset{n}{\sum}x_{n}^m)+p^{m}\right]
+x^n_{n}D(\underset{n}{\sum}x_{n}^n), 
$$
\label{cost2}
where $\textbf{p}$ is the vector (size $N-1$) of the prices for all CP except $n$. The revenue of the CP $m$ is now:
$$
\Pi^v(p^{m},p^{-m})=p^{m}\displaystyle{\underset{n\neq m}{\sum}}\underline{x}_{n}^{m} (\mathbf{p}).
$$

Let $y^n_n$ be the traffic requested by a subscriber connected to an ISP $n$ from the CP $n$ associated to that ISP. 
Let $y^m_n$, with $n\neq m$, traffic requested by a subscriber connected to an ISP $n$ from CP $m$ not associated to that ISP. \\

\begin{figure}[!ht]
\center
\includegraphics[width=8.5cm]{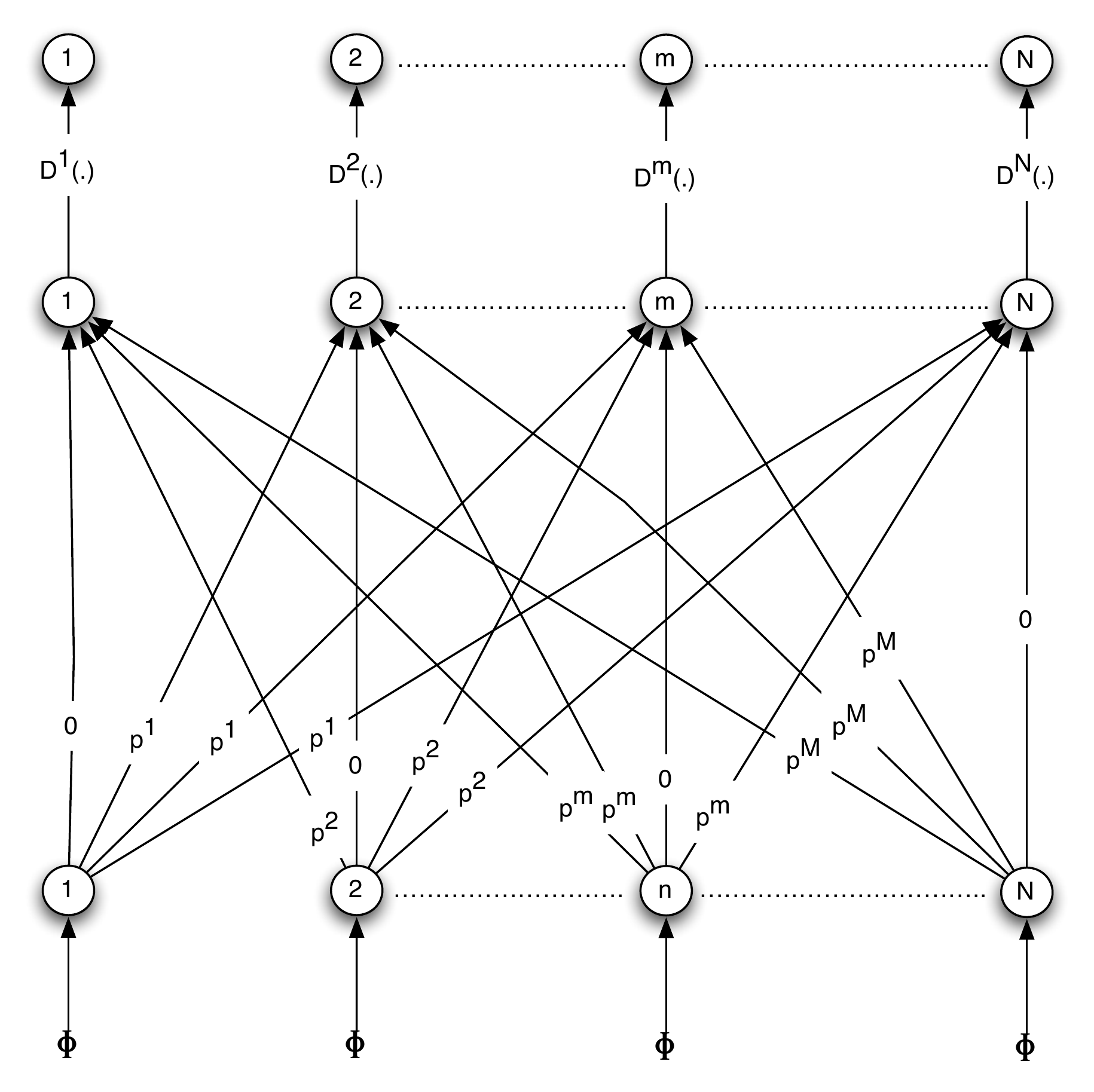}
\label{fig5}
\caption{Equivalent routing game with agreement for M=N}
\end{figure}

\begin{proposition}
In the game with an agreement between each service provider and content provider, it exists for all $(i,n,m)\in \{1,\ldots,I\}\times\{1,\ldots,N\}^2 $ a symmetric equilibrium $(y^{m}_{n},y^{n}_{n},p^m)= (\underline{z},\underline{y},\underline{q})$, which is given by:
 $$
 \underline{q} =a\phi\frac{(N+1)}{3N-1}
$$
and
$$
\underline{z}= \frac{\phi}{N}(\frac{2N-2}{3N-1}), \qquad
\underline{y}= \frac{\phi}{N}(1+\frac{(N-1)(N+1)}{3N-1}).
$$\\
\end{proposition}

\begin{proof}
In order to compute an equilibrium for the game with agreements, we can use the method described previously. \newline
We consider the following Lagrangian function:
$$L^v_{n}(\mathbf{x_{n}},\mathbf{x_{-n}},\mathbf{p},\lambda_{n})=$$ $$\displaystyle{\underset{m\neq n}{ \sum}}x_{n}^{m}\left[D^m(\displaystyle{\underset{n}{\sum}}x_{n}^m)+p^{m}\right]+x_{n}^n[D(\underset{n}{n})x_{n}^n] -\lambda_{n}(\displaystyle{\underset{m}{\sum}}x_{n}^m-\phi).$$

As previously, we assume that CPs $m'\in \{1,\ldots,N \}-\{m\}$ play $q$ and the CP $m$ plays $p^m$. Now again there are several symmetries: we can see that there are two types of subscribers. Each subscriber of each type are interchangeable. Type 1 is subscribers with an agreement with CP m. Type 2 is subscribers without an agreement with CP m.\\
The variables of Type 1 subscribers are:
\begin{itemize}
\item $x$ is the flow from CP $m'$, 
\item $y$ is the flow from CP $m$.
\end{itemize}
The variables of Type 2 subscribers are:
\begin{itemize}
\item  $u$ is the flow from CP $m$,
\item  $v$ is the flow from the CP with has an agreement,
\item  $w$ is the flow from all the other CP except CP $m$ and CP with the agreement.

\end{itemize}
All those decision variables are depicted in figure \ref{fig6}.2

\begin{figure}[!ht]
\center
\includegraphics[width=8.5cm]{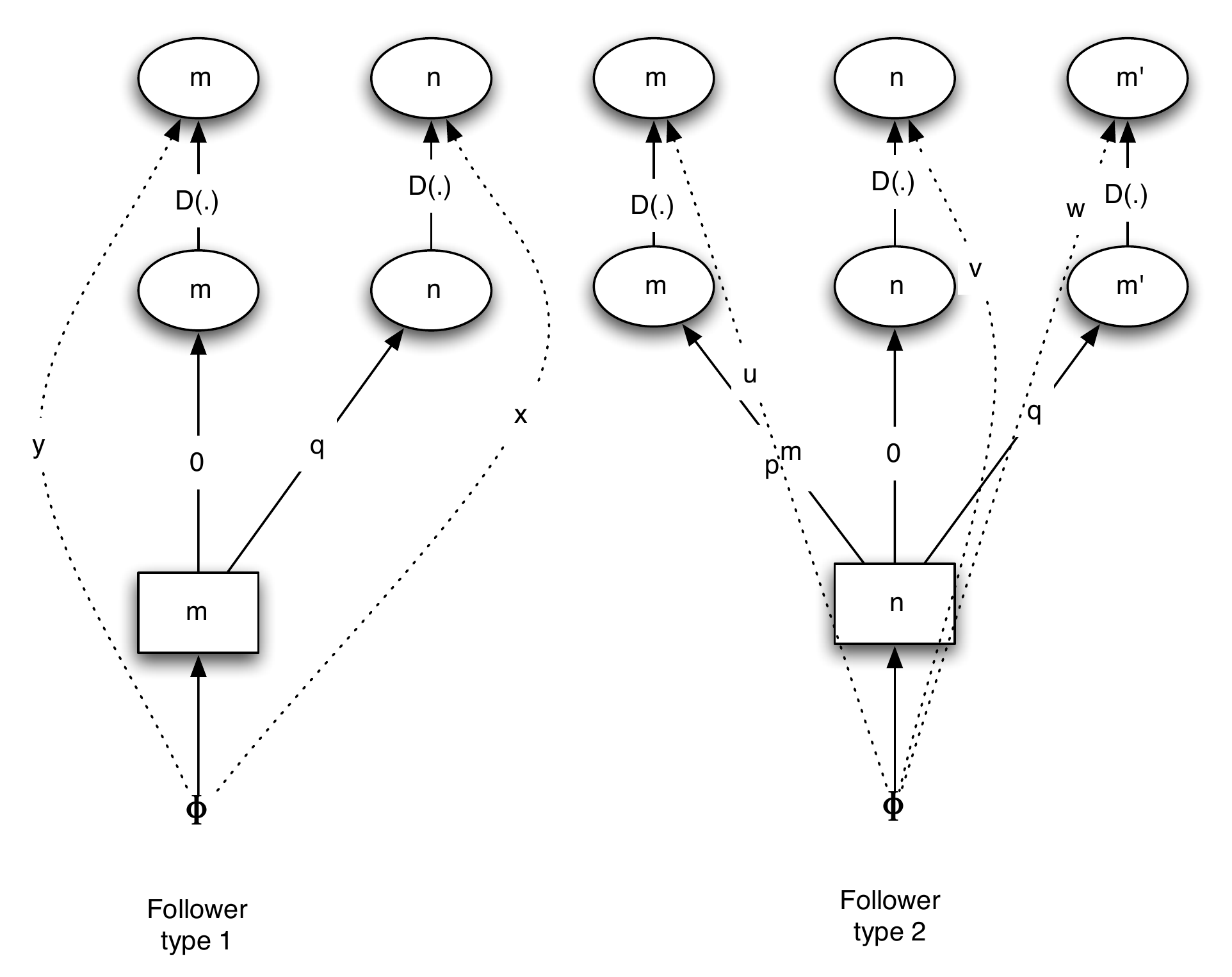}
\label{fig6}
\caption{Followers' strategies in the case of an agreement (help for the proof)}
\end{figure}

Again, thanks to \cite{key-3}, the system is equivalent to the following one with 5 variables $(x,y,u,v,w)$:
\small
$$\begin{cases}
x\displaystyle{\frac{\partial D}{\partial x}}(x+v+(N-2)w)+D(x+v+(N-2)w)+q-\lambda_1=0\\
\\
v\displaystyle{\frac{\partial D}{\partial v}}(x+v+(N-2)w)+D(x+v+(N-2)w)-\lambda_2=0\\
\\
w\displaystyle{\frac{\partial D}{\partial w}}(x+v+(N-2)w)+D(x+v+(N-2)w)+q-\lambda_2=0\\
\\
y\displaystyle{\frac{\partial D}{\partial y}}(y+(N-1)u)+D(y+(N-1)u)-\lambda_1=0\\
\\
u\displaystyle{\frac{\partial D}{\partial u}}(y+(N-1)u)+D(y+(N-1)u)+p^m-\lambda_2=0\\
\\
(M-1)x+y=\phi.\\
\\
(M-2)w+u+v=\phi
\end{cases}$$
\normalsize

 Considering the linear cost function $D(x)=ax$. We have to solve the linear system that are given below:
 
 \small
 $$\begin{cases}
 ax+a(x+v+(N-2)w)+q-\lambda_1=0\\
 \\
 av+a(x+v+(N-2)w)-\lambda_2=0\\
 \\
 aw+a(x+v+(N-2)w)+q-\lambda_2=0\\
 \\
 ay+a(y+(N-1)u)-\lambda_1=0\\
 \\
 au+a(y+(N-1)u)+p^m-\lambda_2=0\\
 \\
 (M-1)x+y=\phi.\\
 \\
 (M-2)w+u+v=\phi
 \end{cases}$$
 \normalsize

We denote $\underline{x}$, $\underline{y}$, $\underline{u}$, $\underline{v}$, $\underline{w}$ the solution of the previous system.  The revenue of CP $m$ is $R^m_v(p^m,q)=p^m\times (N-1) \underline{u}$. To compute the equilibrium price we need to find $p$ which solves $\displaystyle{\frac{\partial R^m_v}{\partial p^m}}(p,p)=0$.
 If $\underline{q}>p_{max}$ then $\underline{q}=p_{max}$. We have a particular interest in the case where $p_{max}>a\phi\frac{(N+1)}{3N-1}$. Equilibrium price is:
$$
 \underline{q} =a\phi\frac{(N+1)}{3N-1}
$$

Then the equilibrium flow from CP $n$ to ISP $m$, $n\neq m$ is \small $\underline{z}= \frac{\phi}{N}(\frac{2N-2}{3N-1})$, and the equilibrium flow from CP $n$ to ISP $n$ is \small\\ $\underline{y}= \frac{\phi}{N}(1+\frac{(N-1)(N+1)}{3N-1}).$\\
\end{proof}

The cost for the subscriber connected to ISP $n$, at the equilibrium, is given by:
$$
C^v_{n}(\underline{y},\underline{z},\underline{q})= \phi^2 a +2a\phi^{2}(\frac{N-1}{3N-1})^2 (\frac{N+1}{N}).
$$
The reward for CP $m$ at the equilibrium is:   
$$
\Pi_m^v(\underline{q},\underline{y},\underline{z})=2a\phi^{2}(\frac{N-1}{3N-1})^2 (\frac{N+1}{N}). 
$$

One important remark is that the download traffic from the privileged CP, $\underline{y}$, has a bounded limit of $\frac{\phi}{3}$ when the number of provider $N$ tends to infinite. In the context without agreements, all the download rates converge to 0. Thus, it means that by making agreement, each CP has a minimum quantity guarantee of traffic to send. It is  an important result for dimensioning CP network infrastructure. In the next section, we compare the cost for any subscriber and the content provider's revenue depending on the economic context without or with agreements between providers.

\section{Comparisons}
\label{comp}
This paper aimed to study the impact of the agreements between service and content providers on the cost of the subscribers and also on the content provider's revenue. In the last two sections we have been able to determine explicitly the equilibrium solutions of the hierarchical game in the two scenarios. Consequently we have obtained the cost of the users and the reward of the CP at this equilibrium. Then, we can compare them to study the impact of those agreements. \\
 
\begin{lem} At equilibrium, the agreement between service and content providers is good for the subscribers, i.e. for all subscribers connected to service provider $n$ we have:
$$
C_n^v(\underline{y},\underline{z},\underline{q})<C_n(\underline{x},\underline{p}).
$$\\
\end{lem}

\begin{proof}
We have the following cost functions to compare: the cost $C_n(\underline{x},\underline{p})$ in the architecture without agreements between providers and $C_n^v(\underline{y},\underline{z},\underline{q})$ the cost for the same subscriber in the architecture with agreements. We have obtained the following expressions depending on the parameters of our system:
$$
C_n(\underline{x},\underline{p})=\phi^2aN
$$
and
$$
C_n^v(\underline{y},\underline{z},\underline{q})= \phi^2 a +2a\phi^{2}(\frac{N-1}{3N-1})^2 (\frac{N+1}{N}).
$$
After some linear algebras (the difference between the costs are equivalent to a third degree polynom), we obtained that for all $N$, $a$ and $\phi$, $C_n^v(\underline{y},\underline{z},\underline{q})<C_n(\underline{x},\underline{p})$.\\
\end{proof}

We first observe that the agreement is beneficial for the users. Nonetheless, agreement decisions are taken by the providers such that they take some benefits by doing this. In a second step, we compare the reward of a content provider depending on if there is agreement or not in the system. We get the following result:\\

\begin{lem} At equilibrium, the agreement between service and content providers is not beneficial for the content provider, i.e. for any content provider $m$ we have:
$$
\Pi_m^v(\underline{q},\underline{y},\underline{z})<\Pi_m(\underline{p},\underline{x}).
$$\\
\end{lem}

\begin{proof}
We have the following expressions of the provider revenue, first without agreements:
$$
\Pi_m(\underline{p},\underline{x})=(N-1)\phi^2a,
$$
and second, with the agreements:
$$
\Pi_m^v(\underline{q},\underline{y},\underline{z})=2a\phi^{2}(\frac{N-1}{3N-1})^2 (\frac{N+1}{N}). 
$$
After some linear algebras and using the result of the previous lemma, we obtained that for all $N$, $a$ and $\phi$, $\Pi_m^v(\underline{q},\underline{y},\underline{z})<\Pi_m(\underline{p},\underline{x})$.\\
\end{proof}

This lemma implies several remarks. First, the agreements should have a positive impact for the service provider. Usually the subscriber pays the service provider for the access, then the agreement implies minimum of traffic from each CP and thus a service provider has an economic interest to make an agreement with a CP. Second, we can change our model by assuming that each CP charges $c$ instead of 0 for the traffic to the privileged ISP and $\underline{q}+c$ to the other  ISP. This change does not bring any modification in the equilibrium of our hierarchical game and then we can determine the value $c$ such that the revenue of each CP is the same with or without an agreement. Therefore, both the subscribers and the providers will gain by introducing such agreements in the market.

\section{Conclusions and perspectives}
\label{conc}

In this paper, we have studied the impact of a pricing agreement between service and content providers on the Internet users. We have evaluated this impact by modeling the system with a hierarchical game in which Internet users, called subscribers, split their demand (download traffic) from several content providers depending on costs. Those costs depend on the preference (which depends on the congestion at each CP) and an access price paid to the content provider. At the upper-layer of the hierarchical game, the content providers compete through their prices in order to maximize their own revenue.  Our first main result is that introducing agreements between service and content providers, causes an impact on the cost perceived by the Internet users. Therefore, it also brings an impact to the revenue of the content providers. In fact, we have proved that such agreements have a positive effect for the Internet users and a negative effect on the content providers. Based on these interesting results, we can think about several extensions. First, we can also introduce quality of service controlled by the content provider and/or the service provider. Second we can do again all our computation in the case of an elastic demand.

\end{document}